\documentclass[10pt]{amsart}
\usepackage[leqno]{amsmath}
\usepackage{amsthm}
\usepackage{amsfonts}
\usepackage{amssymb}
\usepackage{eucal}
\usepackage[all]{xy}
\usepackage{pseudocode}
\CompileMatrices

\DeclareFontFamily{OT1}{rsfs}{}
\DeclareFontShape{OT1}{rsfs}{n}{it}{<-> rsfs10}{}
\DeclareMathAlphabet{\mathscr}{OT1}{rsfs}{n}{it}

\theoremstyle{plain}
  \newtheorem{theorem}[subsubsection]{Theorem}
  \newtheorem{proposition}[subsubsection]{Proposition}
  \newtheorem{lemma}[subsubsection]{Lemma}

\theoremstyle{definition}

\theoremstyle{remark}

\numberwithin{equation}{subsection}

\title{Fast matrix multiplication techniques based on the Adleman-Lipton model}
\author{Aran Nayebi}
\email{aran.nayebi@gmail.com}
\keywords{DNA computing; residue number system; logic and arithmetic operations; Strassen algorithm}
\subjclass[2010]{Primary 65F05, 03D10; Secondary 68Q10, 68Q05, 03D80}
\begin{document}
\begin{abstract}
On distributed memory electronic computers, the implementation and association of fast parallel matrix multiplication algorithms has yielded astounding results and insights. In this discourse, we use the tools of molecular biology to demonstrate the theoretical encoding of Strassen's fast matrix multiplication algorithm with DNA based on an $n$-moduli set in the residue number system, thereby demonstrating the viability of computational mathematics with DNA. As a result, a general scalable implementation of this model in the DNA computing paradigm is presented and can be generalized to the application of \emph{all} fast matrix multiplication algorithms on a DNA computer. We also discuss the practical capabilities and issues of this scalable implementation. Fast methods of matrix computations with DNA are important because they also allow for the efficient implementation of other algorithms (i.e. inversion, computing determinants, and graph theory) with DNA.
\end{abstract}

\maketitle

\section{Introduction}
The multiplication of matrices is a fundamental operation applicable to a diverse range of algorithms from computing determinants, inverting matrices, and solving linear systems to graph theory. Indeed, Bunch and Hopcroft \cite{bunch-hopcroft} successfully proved that given an algorithm for multiplying two $n \times n$ matrices in $O(n^{\alpha})$ operations where $2 < \alpha \le 3$, then the triangular factorization of a permutation of any $n \times n$ nonsingular matrix as well as its inverse can be found in $O(n^{\alpha})$ operations. The standard method of square matrix multiplication requires $2n^3$ operations. Let $\omega$ be the smallest number such that $O(n^{\omega+\epsilon})$ multiplications suffice for all $\epsilon > 0$. Strassen \cite{strassen} presented a divide-and-conquer algorithm using noncommutative multiplication to compute the product of two matrices (of order $m2^k$) by $m^37^k$ multiplications and $(5+m)m^27^k-6m^22^{2k}$ additions. Thus, by recursive application of Strassen's algorithm, the product of two matrices can be computed by at most $(4.7)n^{\log_2 7}$ operations. Following Strassen's work, Coppersmith and Winograd \cite{coppersmith-winograd} were able to improve the exponent to 2.38. Their approaches and those of subsequent researchers rely on the same framework: For some $k$, they devise a method to multiply matrices of order $k$ with $m \lll k^3$ multiplications and recursively apply this technique to show that $\omega < \log_k m$ \cite{robinson}. Only until recently, it was long supposed that $\omega$ could take on the value of 2 without much evidence. Using a group-theoretic construction, Cohn, Kleinberg, Szegedy, and Umans \cite{exponent2} rederived the Coppersmith-Winograd algorithm to describe several families of wreath product groups that yield nontrivial upper bounds on $\omega$, the best asymptotic result being 2.41. They also presented two conjectures in which either one would imply an exponent of 2. \newline
\indent Unfortunately, although these improvements to Strassen's algorithm are theoretically optimal, they lack pragmatic value. In practice, only the Strassen algorithm is fully implemented and utilized as such: \newline
\indent \indent For even integers $m$, $n$, and $k$, let $X \in {\mathbb{R}}^{m \times k}$ and $Y \in {\mathbb{R}}^{k \times n}$ be matrices with product $Q \in {\mathbb{R}}^{m \times n}$, and set
\begin{equation*}
X = \begin{pmatrix} X_{00} & X_{01} \\ X_{10} & X_{11} \end{pmatrix}, \indent Y = \begin{pmatrix} Y_{00} & Y_{01} \\ Y_{10} & Y_{11} \end{pmatrix}, \indent Q = \begin{pmatrix} Q_{00} & Q_{01} \\ Q_{10} & Q_{11} \end{pmatrix},
\end{equation*}
where $X_{ij} \in {\mathbb{R}}^{m/2 \times k/2}$, $Y_{ij} \in {\mathbb{R}}^{k/2 \times n/2}$, and $Q_{ij} \in {\mathbb{R}}^{m/2 \times n/2}$. Then perform the following to compute $Q = XY$,
\begin{equation*}
{M}_{0} := ({X}_{00} + {X}_{11}) ({Y}_{00} + {Y}_{11}),
\end{equation*}
\begin{equation*}
{M}_{1} := ({X}_{10} + {X}_{11}){Y}_{00},
\end{equation*}
\begin{equation*}
{M}_{2} := {X}_{00} ({Y}_{01} - {Y}_{11}),
\end{equation*}
\begin{equation*}
{M}_{3} := {X}_{11} (-{Y}_{00}+{Y}_{10}),
\end{equation*}
\begin{equation*}
{M}_{4} := ({X}_{00} + {X}_{01}){Y}_{11},
\end{equation*}
\begin{equation*}
{M}_{5} := (-{X}_{00}+{X}_{10}) ({Y}_{00} + {Y}_{01}),
\end{equation*}
\begin{equation*}
{M}_{6} := ({X}_{01} - {X}_{11}) ({Y}_{10} + {Y}_{11}),
\end{equation*}
\begin{equation*} 
{Q}_{00} = {M}_{0} + {M}_{3} - {M}_{4} + {M}_{6},
\end{equation*}
\begin{equation*}
{Q}_{01} = {M}_{1} + {M}_{3},
\end{equation*}
\begin{equation*}
{Q}_{10} = {M}_{2} + {M}_{4},
\end{equation*}
\begin{equation*}
{Q}_{11} = {M}_{0} + {M}_{2} - {M}_{1} + {M}_{5}.
\end{equation*}
Even if the dimension of the matrices is not even or if the matrices are not square, it is easy to pad the matrices with zeros and perform the aforementioned algorithm. \newline 
\indent Typically, computations such as this one are performed using electronic components on a silicon substrate. In fact, it is a commonly held notion that \emph{most} computers should follow this model. In the last decade however, a newer and more revolutionary form of computing has come about, known as DNA computing. DNA's key advantage is that it can make computers much smaller than before, while at the same time maintaining the capacity to store prodigious amounts of data. Since Adleman's \cite{adleman} pioneering paper, DNA computing has become a rapidly evolving field with its primary focus on developing DNA algorithms for NP-complete problems. However, unlike quantum computing in recent years, the viability of computational mathematics on a DNA computer has not yet been fully demonstrated, for the whole field of DNA-based computing has merged to controlling  and mediating information processing for nano structures and molecular movements. In fact, only recently have the primitive operations in mathematics (i.e. addition, subtraction, multiplication, and division) been implemented. Thus, the general problem dealt with in this paper is to explore the feasibility of computational mathematics with DNA. Fujiwara, Matsumoto, and Chen \cite{fujiwara} proved a DNA representation of binary integers using single strands and presented procedures for primitive mathematical operations through simple manipulations in DNA. It is important to note that the work of Fujiwara et al. \cite{fujiwara} and those of subsequent researchers have relied upon a fixed-base number system. The fixed-base number system is a bottleneck for many algorithms as it restricts the speed at which arithmetic operations can be performed and increases the complexity of the algorithm. Parallel arithmetic operations are simply not feasible in the fixed-base number system because of the effect of a carry propagation. Recently, Zheng, Xu, and Li \cite{zheng} have presented an improved DNA representation of an integer based on the residue number system (RNS) and give algorithms of arithmetic operations in $Z_M = \{0,1,\cdots,M-1\}$ where $Z_M$ is the ring of integers with respect to modulo $M$. Their results exploit the massive parallelism in DNA mainly because of the carry-free property of all arithmetic operations (except division, of course) in RNS. \newline
\indent In this paper we present a parallelization method for performing Strassen's fast matrix multiplication methods on a DNA computer. Although DNA-based methods for the multiplication of boolean \cite{oliver} and real-numbered matrices \cite{zhang} have been proven, these approaches use digraphs and are not divide-and-conquer like Strassen's algorithm (and hence are not particularly efficient when used with DNA). Divide-and-conquer algorithms particularly benefit from the parallelism of the DNA computing paradigm because distinct sub-processes can be executed on different processors. The critical problem addressed in this paper is to provide a DNA implementation of Strassen's algorithm, while keeping in mind that in recent years it has been shown that the biomolecular operations suggested by the Adleman-Lipton model are not very reliable in practice. More specifically, the objectives we aim to accomplish in this research paper are the following:
\begin{itemize}
\item To provide in \S 2 a revised version of the Adleman-Lipton model that better handles recursive ligation and overcomes the confounding of results with the complexity of tube content.
\item To establish a systematic approach in \S 3 of representing and adding and subtracting matrices using DNA in the RNS system.
\item Next, based on this representation system, we describe in \S 4.1 an implementation of the Cannon algorithm with DNA at the bottom level.
\item And lastly, we present in \S 4.2 a method to store the different sub-matrices in different strands, and in \S 4.3, we prove a mathematical relation between the resultant matrix and the sub-matrices at recursion level $r$.
\end{itemize}
Our approach uses the Cannon algorithm at the bottom level (within a tube containing a memory strand) and the Strassen algorithm at the top level (between memory strands). We show that the Strassen-Cannon algorithm decreases in complexity as the recursion level $r$ increases \cite{algo}. If the Cannon algorithm is replaced by other parallel matrix multiplication algorithms at the bottom level (such as the Fox algorithm), our result still holds. The difficulty that arises is that in order to use the Strassen algorithm at the top level, we must determine the sub-matrices after the recursive execution of the Strassen formula $r$ times and then find the resultant matrix. On a sequential machine, this problem is trivial; however, on a parallel machine this situation becomes much more arduous. Nguyen, Lavall\'{e}e, and Bui \cite{algo} present a method for electronic computers to determine all the nodes at the unspecified level $r$ in the execution tree of the Strassen algorithm, thereby allowing for the direct calculation of the resultant matrix from the sub-matrices calculated by parallel matrix multiplication algorithms at the bottom level. Thus, we show that this result can theoretically be obtained using DNA, and combined with a storage map of sub-matrices to DNA strands and with the usage of the Cannon algorithm at the bottom level, we have a general scalable implementation of the Strassen algorithm on Adleman's DNA computer. As of the moment, we should note that this implementation is primarily theoretical because in practice, the Adleman-Lipton model is not always feasible, as explained in \S 5. The reason why we concentrate on the Strassen algorithm is that it offers superior performance than the traditional algorithm for practical matrix sizes less than $10^{20}$ \cite{algo}. However, our methods are also applicable to \emph{all} fast matrix multiplication algorithms on a DNA computer, as these algorithms are always in recursive form \cite{pan}. In addition, our results can be used to implement other algorithms such as inversion and computing determinants on a DNA computer since matrix multiplication is almost ubiquitous in application.
\section{Preliminary Theory}
\subsection{The Residue Number System}
The residue number system is defined by a set of pairwise, coprime moduli $P = \{q_{n-1},\cdots,q_0\}$. Furthermore, an integer in RNS is represented as a vector of residues with respect to the moduli set $P$. As a consequence of the Chinese remainder theorem, for any integer $x \in \left[0, M-1\right]$ where $M = \prod_{i=0}^{n-1}q_{i}$, each RNS representation is unique. As stated by Zheng, Xu, and Li \cite{zheng}, the vector $(x_{n-1},\cdots,x_{0})$ denotes the residue representation of $x$. \newline
\indent It has been previously mentioned that one of the important characteristic of RNS is that all arithmetic operations except for division are carry-free. Thus, for any two integers $x \to (x_{n-1},\cdots,x_{0})\in Z_{M}$ and $y \to (y_{n-1},\cdots,y_{0})\in Z_{M}$ we obtain the following from \cite{residue}:
\begin{equation} \label{1}
|x \circ y|_{M} \to \left(|x_{n-1}\circ y_{n-1}|_{q_{n-1}},\cdots,|x_0\circ y_0|_{q_{0}}\right),
\end{equation}
in which $\circ$ is any operation of addition, subtraction, or multiplication.
\subsection{The Adleman-Lipton Model}
In this section we present a theoretical and practical basis for our algorithms. By the Adleman-Lipton model, we define a test tube $T$ as a multi-set of (oriented) DNA sequences over the nucleotide alphabet $\{A,G,C,T\}$. The following operations can be performed as follows: \newline
\begin{itemize}
\item $Merge(T_1,T_2)$: merge the contents in tube $T_1$ and tube $T_2$, and store the results in tube $T_1$;
\item $Copy(T_1,T_2)$: make a copy of the contents in tube $T_1$ and store the result in tube $T_2$;
\item $Detect(T)$: for a given tube $T$, this operation returns ``True'' if tube $T$ contains \emph{at least} one DNA strand, else it returns ``False'';
\item $Separation(T_1,X,T_2)$: from all the DNA strands in tube $T_1$, take out only those containing the sequences of $X$ over the alphabet $\{A,G,C,T\}$ and place them in tube $T_2$;
\item $Selection(T_1,l,T_2)$: remove all strands of length $l$ from tube $T_1$ into tube $T_2$;
\item $Cleavage(T,\sigma_0\sigma_1)$: given a tube $T$ and a sequence $\sigma_0\sigma_1$, for every strand containing $\begin{bmatrix}\sigma_0\sigma_1 \\ \overline{\sigma_0\sigma_1}\end{bmatrix}$, then the cleavage operation can be performed as such:
\begin{equation*}
\begin{bmatrix}\alpha_0\sigma_0\sigma_1\beta_0 \\ \alpha_1\overline{\sigma_0\sigma_1}\beta_1\end{bmatrix} \xrightarrow{Cleavage(T,\sigma_0\sigma_1)}\begin{bmatrix}\alpha_0\sigma_0 \\ \alpha_1\overline{\sigma_0}\end{bmatrix}, \indent \begin{bmatrix}\sigma_1\beta_0 \\ \overline{\sigma_1}\beta_1\end{bmatrix},
\end{equation*}
where the overhead bar denotes the complementary strand.
\item $Annealing(T)$: produce \emph{all} feasible double strands in tube $T$ and store the results in tube $T$ (the assumption here is that ligation is executed after annealing);
\item $Denaturation(T)$: disassociate every double strand in tube $T$ into two single strands and store the results in tube $T$;
\item $Empty(T)$: empty tube $T$.
\end{itemize}
According to \cite{residue}, the complexity of each of the aforementioned operations is $O(1)$.

\subsection{Revised Adleman-Lipton Model through Ligation by Selection}
In practice, the recursive properties of our implementation of the Strassen-Canon algorithm require a massive ligation step that is not feasible. The reason is that, in practice, the biomolecular operations suggested by the Adleman-Lipton model are not completely reliable. This ligation step cannot produce longer molecules as required by our implementation, and certainly not more than 10-15 ligations in a row. Not to mention that both the complexity of the tube content and the efficiency of the enzyme would obscure the results. As a result of these considerations, the operations $Separation(T1, X, T2)$ and $Annealing(T)$ presented in \S 2.2 function with questionable success when applied to a complex test tube, especially when recursion is used. \newline
\indent Therefore, in order for matrix multiplication under the Adleman-Lipton model to be completely reliable in practice and the aforementioned problems circumvented, these streptavidin based operations must be improved upon. That way, the parallelization offered by DNA can be utilized as an important mathematical tool with performance capabilities comparable to the electronics.
One way we propose to overcome this potential setback of ligation is to use a modified ligation procedure that can handle longer molecules in place of the original, termed "ligation by selection" presented in \cite{kodumal}. Ligation by selection (LBS) is a method to ligate multiple adjacent DNA fragments that does not require intermediate fragment isolation and is amenable to parallel processing, therefore reducing the obfuscation of the results by the complexity of tube content. Essentially in LBS, fragments that are adjacent to each other are cloned into plasmid markers that have a common antibiotic marker, a linking restriction site for joining the fragments, a linking restriction site on the vector, and each vector has a unique site to be used for restriction-purification and a unique antibiotic marker. The method is applied to efficiently stitch multiple synthetic DNA fragments of 500-800 bp together to produce segments of up to 6000 bp in \cite{kodumal}. For a cogent and complete explanation of ligation by selection we refer the reader to \cite{kodumal}. \newline
\indent To utilize LBS recursively, the alteration of resistance markers and restriction-purification sites of acceptor and donor vectors that occur in each LBS cycle must be accounted for in order to minimize the number of cycles required in parallel processing. As opposed to conventional ligation, the advantages that LBS has are \cite{kodumal}:
\begin{itemize}
\item The avoidance of the need to isolate, purify, and ligate individual fragments
\item The evasion of the need for specialized MCS linkers
\item And most importantly, the ease with which parallel processing of operations may be applied
\end{itemize}
Hence, in order for the Adleman-Lipton model to be more relaible in the recursive operations our implementation of Strassen's algorithm requires, we replace the ligation procedure of \S 2.2 with LBS.

\section{DNA Matrix Operations in RNS}
\subsection{DNA Representation of a Matrix in RNS}
We extend the DNA representation of integers in RNS presented in \cite{zheng} to representing an entire matrix $Y$ in RNS by way of single DNA strands. Let matrix $Y$ be a $t \times t$ matrix with:
\footnotesize{
\begin{equation*}
Y = \begin{pmatrix}y_{11} & y_{12} & \cdots & y_{1t} \\ y_{21} & y_{22} & \cdots & y_{2t} \\ \vdots & \vdots & \ddots & \vdots \\ y_{t1} & y_{t2} & \cdots & y_{tt}\end{pmatrix}.
\end{equation*}}
\normalsize{
The key here is the RNS representation of each element $y_{qr}$ in the hypothetical matrix $Y$ with $1 \le q \le t$ and $1 \le r \le t$ by way of DNA strands. \newline
\indent We first utilize the improved DNA representation of $n$ binary numbers with $m$ binary bits as described in \cite{zheng} for the alphabet $\sum$:
\footnotesize{
\begin{equation*}
\sum = \{A_i,B_j,C_0,C_1,E_0,E_1,D_0,D_1,1,0,\#|0\le i \le M-1,0\le j \le m\}.
\end{equation*}}
\normalsize{
Here, $A_i$ indicates the address of $M$ integers in RNS; $B_j$ denotes the binary bit position; $C_0$, $C_1$, $E_0$, $E_1$, $D_0$, and $D_1$ are used in the $Cleavage$ operation; $\#$ is used in the $Separation$ operation; and 0 and 1 are binary numbers. Thus, in the residue digit position, the value of the bit $y_{qr}$ with a bit address of $i$ and a bit position of $j$ can be represented by a single DNA strand $(S_{i,j}){y_{qr}}$:
}
\footnotesize{
\begin{equation} \label{2}
(S_{i,j})_{qr} = (D_1B_jE_0E_1A_iC_0C_1VD_0){y_{qr}},
\end{equation}}
\normalsize{
for $V \in \{0,1\}$. Hence, the matrix $Y$ can be represented as such:
}
}
\footnotesize{
\begin{equation*}
Y = \begin{pmatrix}(D_1B_jE_0E_1A_iC_0C_1VD_0)_{y_{11}} & (D_1B_jE_0E_1A_iC_0C_1VD_0)_{y_{12}} & \cdots & (D_1B_jE_0E_1A_iC_0C_1VD_0)_{y_{1t}} \\ (D_1B_jE_0E_1A_iC_0C_1VD_0)_{y_{21}} & (D_1B_jE_0E_1A_iC_0C_1VD_0)_{y_{22}} & \cdots & (D_1B_jE_0E_1A_iC_0C_1VD_0)_{y_{2t}} \\ \vdots & \vdots & \ddots & \vdots \\ (D_1B_jE_0E_1A_iC_0C_1VD_0)_{y_{t1}} & (D_1B_jE_0E_1A_iC_0C_1VD_0)_{y_{t2}} & \cdots & (D_1B_jE_0E_1A_iC_0C_1VD_0)_{y_{tt}}\end{pmatrix},
\end{equation*}}
\normalsize{
where each strand-element is not necessarily distinct. The reader must keep in mind that $M$ integers in RNS defined by the $n$-moduli set $P$ can be represented by $2M(m+1)$ different memory strands, whereas in the binary system, the respresentation of $M$ integers requires $2M\left(1+\sum_{i=0}^{n-1}m_i\right)$ different memory strands.
}

\normalsize{
\subsection{Residue Number Arithmetic with Matrices}
From \eqref{1}, it is apparent that the operation $\circ$ is carry-free, thereby allowing for the employment of parallel procedures in all residue digits. In \cite{zheng} two properties are given for the modular operation involving two integers $x \to \left(x_{n-1},\cdots,x_0\right)$ and $y \to \left(y_{n-1},\cdots,y_0\right)$ in RNS defined by the set $P = \{2^{m_{n-1}},2^{m_{n-2}}-1,\cdots,2^{m_0}-1\}$.
\begin{lemma} \label{property1}
For $\forall j$, $m_{n-1} \in \mathbb{N}$, if $j < m_{n-1}$ then $|2^{j}|_{2^{m_{n-1}}} = 2^j$ else $|2^{j}|_{2^{m_{n-1}}} = 0$.
\end{lemma}
\begin{lemma} \label{property2}
For $l = 0,\cdots,n-2$, let $x_l+y_l = z_l$ where $z_l = (z_{l(m_l)},\cdots,z_{l0})$. If $z_l > 2^{m_l}-1$, then $|z_l|_{2^{m_l}-1} = 1 + \sum_{j=0}^{m_l-1}z_{lj}2^{j}$.
\end{lemma}
\indent Next, the procedures \emph{RNSAdd} and \emph{RNSDiff} add and subtract two integers in RNS defined by the moduli set $P$, respectively. The pseudocode for \emph{RNSAdd} and \emph{RNSDiff} is given in \S 4.4 of \cite{zheng}, and we refer the reader to that source (note that the pseudocode of \cite{zheng} for both algorithms utilizes the operations presented in \S 2.2 extensively). Instead, we provide some background on the two procedures. The inputs are $2n$ tubes $T_{l}^{x_{qr}}$ and $T_{l}^{y_{qr}}$ (for $l = 0, \cdots, n-1$) containing the memory strands representing the elements $x_{qr}$ and $y_{qr}$ of $t \times t$ matrices $X$ and $Y$, respectively. Once either operation is complete, it returns $n$ tubes $T_{l}^{Rsum}$ and $T_{l}^{Rdiff}$ containing the result of residue addition or subtraction, respectively. We also use the following $n$ temporary tubes for \emph{RNSAdd}, namely, $T_{temp}^{l}$, $T_{sum}^{l}$, and $T_{sum'}^{l}$. Similarly for \emph{RNSDiff}, the $n$ temporary tubes, $T_{temp}^{l}$, $T_{diff}^{l}$, and $T_{diff'}^{l}$ are used. \newline
\indent Thus, based on Lemma~\ref{property1} and Lemma~\ref{property2}, we introduce the following two algorithms for matrix addition and subtraction in RNS which will be used when dealing with the block matrices in Strassen's algorithm. For the sake of example, we are adding (and subtracting) the hypothetical $t \times t$ matrices $X$ and $Y$. Essentially, the \emph{RNSMatrixAdd} and \emph{RNSMatrixDiff} algorithms employ \emph{RNSAdd} and \emph{RNSDiff} in a nested FOR loop.
\subsubsection{Matrix Addition}
\indent The procedure \emph{RNSMatrixAdd} is defined as: \newline
\footnotesize{\begin{pseudocode}{RNSMatrixAdd}{T_X,T_Y}
\FOR q \GETS 1 \TO t \DO \\
\BEGIN
	\FOR r \GETS 1 \TO t \DO \\
		\BEGIN
		\text{RNSAdd}(T_{n-1}^{x_{qr}},\cdots,T_{0}^{x_{qr}},T_{n-1}^{y_{qr}},\cdots,T_{0}^{y_{qr}});
		\END\\
\END\\
\end{pseudocode}}
\normalsize{\subsubsection{Matrix Subtraction}
\indent The procedure \emph{RNSMatrixDiff} is defined as: \newline
\footnotesize{\begin{pseudocode}{RNSMatrixDiff}{T_X,T_Y}
\FOR q \GETS 1 \TO t \DO \\
\BEGIN
	\FOR r \GETS 1 \TO t \DO \\
		\BEGIN
		\text{RNSDiff}(T_{n-1}^{x_{qr}},\cdots,T_{0}^{x_{qr}},T_{n-1}^{y_{qr}},\cdots,T_{0}^{y_{qr}});\\
		\END\\
\END\\
\end{pseudocode}}
\normalsize{\section{Strassen's Algorithm Revisited}
\subsection{Bottom-Level Matrix Multiplication}
Although a vast repository of traditional matrix multiplication algorithms can be used between processors (or in our case, test tubes containing memory strands; however for the sake of brevity, we shall just use the term ``memory strand'' or ``strand''), we will employ the Cannon algorithm \cite{cannon} since it can be used on matrices of any dimension. We will only discuss square strand arrangments and square matrices for simplicity's sake. Assume that we have $p^2$ memory strands, organized in a logical sequence in a $p \times p$ mesh. For $i \ge 0$ and $j \le p-1$, the strand in the $i^{\text{th}}$ row and $j^{\text{th}}$ column has coordinates $(i,j)$. The matrices $X$, $Y$, and their matrix product $Q$ are of size $t \times t$, and again as a simplifying assumption, let $t$ be a multiple of $p$. All matrices will be partitioned into $p \times p$ blocks of $s \times s$ sub-matrices where $s = t/p$. As described in \cite{algo}, the mesh can be percieved as an amalgamation of rings of memory strands in both the horizontal and vertical directions (opposite sides of the mesh are linked with a torus interconnection). A successful DNA implementation of Cannon's algorithm requires communication between the strands of each ring in the mesh where the blocks of matrix $X$ are passed \emph{in parallel} to the left along the horizontal rings and the blocks of the matrix $Y$ are passed to the top along the vertical rings. Let $X_{ij}$, $Y_{ij}$, and $Q_{ij}$ denote the blocks of $X$, $Y$, and $Q$ stored in the strand with coordinates $(i,j)$. The Cannon algorithm on a DNA computer can be described as such: \newline
\newline
\footnotesize{
\begin{pseudocode}{Cannon}{T_{X_{ij}},T_{Y_{ij}}}
\FOR i^{\text{th}}\text{ column} \GETS 0 \TO i \DO \\
	\BEGIN
	\text{LeftShift}(T_{X_{ij}})\\
	\END\\
\FOR j^{\text{th}}\text{ column} \GETS 0 \TO j \DO \\
	\BEGIN
	\text{UpShift}(T_{Y_{ij}})\\
	\END\\
\forall \text{ strands } (i,j) \DO\\
	\BEGIN
	\text{ValueAssignment}(T_{X_{ij}Y_{ij}},T_{Q_{ij}})\\
	\END\\
\DO (p-1) \text{ times}\\
	\BEGIN
	\text{LeftShift}(T_{X_{ij}})\\
	\text{UpShift}(T_{Y_{ij}})\\
	\text{ValueAssignment}\left(T_{\text{RNSMatrixAdd}(T_{Q_{ij}},T_{X_{ij}Y_{ij}})},T_{Q_{ij}}\right)\\
	\END\\
\end{pseudocode}
}
\normalsize{
\newline
\indent Note that the procedure \emph{UpShift} can be derived from Zheng et al.'s \cite{zheng} \emph{LeftShift}. Now we examine the run-time of the Cannon algorithm. The run time can be componentized into the communication time and the computation time, and the total communication time is
\begin{equation} \label{5}
2p\alpha + \frac{2B\beta t^2}{p},
\end{equation}
and the computation time is
\begin{equation} \label{6}
\frac{2t^3t_{comp}}{p^2},
\end{equation}
where $t_{comp}$ is the execution time for one arithmetic operation, $\alpha$ is the latency, $\beta$ is the sequence-transfer rate, the total latency is $2p \alpha$, and the total sequence-transfer time is $2p\beta B(m/p)^2$ with $B$ as the number of sequences to store one entry of the matrices. According to \cite{algo}, the running time is
\begin{equation} \label{7}
T(t) = \frac{2t^3t_{comp}}{p^2} + 2p\alpha + \frac{2B\beta t^2}{p}.
\end{equation}
\subsection{Matrix Storage Pattern}
The primary difficulty is to be able to store the different sub-matrices of the Strassen algorithm in different strands, and these sub-matrices must be copied or moved to appropriate strands if tasks are spawned. Hence, we present here a storage map of sub-matrices to strands based on the result of Luo and Drake \cite{scalable} for electronic computers. Essentially, if we allow each strand to have a portion of each sub-matrix at each resursion level, then we can make it possible for all strands to act as \emph{one} strand. As a result, the addition and subtraction of the block matrices performed in the Strassen algorithm at all recursion levels can be performed in parallel without any inter-strand communication \cite{algo}. Each strand performs its local sub-matrix additions and subtractions in RNS (via \emph{RNSMatrixAdd} and \emph{RNSMatrixDiff} described in \S 3). At the final recursion level, the block matrix multiplications are calculated using the Cannon algorithm in \S 4.1. \newline
\indent For instance, if we suppose that the recursion level in the Strassen-algorithm is $r$, and let $n = t/p$, $t_0 = t/2$, and $n_0 = t_0/p$ for $n,t_0,n_0 \in \mathbb{N}$, then the run-time of the Strassen-Canon algorithm is:
\begin{equation} \label{8}
T(t) = 18T_{add}\left(\frac{t}{2}\right) + 7T\left(\frac{t}{2}\right),
\end{equation}
where $T_{add}\left(\frac{t}{2}\right)$ is the run-time to add or subtract block matrices of order $t/2$. \newline
Additionally, according to (9) of \cite{algo},
\begin{equation} \label{9}
T_t \approx \frac{2(\frac{7}{8})^r t^3t_{comp}}{p^2} + \frac{5(\frac{7}{4})^r t_{comp}}{p^2} + \left(\frac{7}{4}\right)^r 2p\alpha.
\end{equation}
Since the asymptotically significant term $\frac{2(\frac{7}{8})^r t^3t_{comp}}{p^2}$ decreases as the recursion level $r$ increases, then for $t$ significantly large, the Strassen-Cannon algorithm should be faster than the Cannon algorithm. Even if the Cannon algorithm is replaced at the bottom level by other parallel matrix multiplication algorithms, the same result holds.
\subsection{Recursion Removal}
As has been previously discussed, in order to use the Strassen algorithm between strands (at the top level), we must determine the sub-matrices after $r$ times recursive execution and then to determine the resultant matrix from these sub-matrices. Nguyen et al. \cite{algo} recently presented a method on electronic computers to ascertain all of the nodes in the execution tree of the Strassen algorithm at the unspecified recursion level $r$ and to determine the relation between the sub-matrices and the resultant matrix at level $r$. We extend it to the DNA computing paradigm. At each step, the algorithm will execute a multiplication between 2 factors, namely the linear combinations of the elements of the matrices $X$ and $Y$, respectively. Since we can consider that each factor is the sum of all elements from each matrix, with coefficient of 0, -1, or 1 \cite{algo}, then we can represent these coefficients with the RNS representation of numbers with DNA strands described in \S 3.1 as such:
\begin{equation*}
\begin{split}
(\{D_1B_1E_0E_1A_0C_0C_10D_0,D_1B_0E_0E_1A_0C_0C_10D_0\},\{D_1B_1E_0E_1A_0C_0C_10D_0,D_1B_0E_0E_1A_0C_0C_10D_0\},\\ \{D_1B_1E_0E_1A_0C_0C_10D_0,D_1B_0E_0E_1A_0C_0C_10D_0\}),
\end{split}
\end{equation*}
\begin{equation*}
\begin{split}
(\{D_1B_1E_0E_1A_{-1}C_0C_11D_0,D_1B_0E_0E_1A_{-1}C_0C_11D_0\},\{D_1B_1E_0E_1A_{-1}C_0C_11D_0,D_1B_0E_0E_1A_{-1}C_0C_11D_0\},\\
\{D_1B_1E_0E_1A_{-1}C_0C_11D_0,D_1B_0E_0E_1A_{-1}C_0C_11D_0\}),
\end{split}
\end{equation*}
or
\begin{equation*}
\begin{split}
(\{D_1B_1E_0E_1A_{1}C_0C_10D_0,D_1B_0E_0E_1A_{1}C_0C_11D_0\},\{D_1B_1E_0E_1A_{1}C_0C_10D_0,D_1B_0E_0E_1A_{1}C_0C_11D_0\},\\
\{D_1B_1E_0E_1A_{1}C_0C_10D_0,D_1B_0E_0E_1A_{1}C_0C_11D_0\}),
\end{split}
\end{equation*}
respectively. For the sake of brevity, we shall denote the latter three equations as $(0)_{RNS}$, $(-1)_{RNS}$, and $(1)_{RNS}$, respectively. This coefficient is obtained for each element in each recursive call and is dependent upon both the index of the call and the location of an element in the division of the matrix by 4 sub-matrices \cite{algo}. If we view the Strassen-Cannon algorithm's execution as an execution tree \cite{algo}, then each scalar multiplication is correlated on a leaf of the execution tree and the path from the root to the leaf represents the recursive calls leading to the corresponding multiplication. Furthermore, at the leaf, the coefficient of each element (either $(0)_{RNS}$, $(-1)_{RNS}$, or $(1)_{RNS}$) can be determined by the combination of all computations in the path from the root. The reason is that since all of the computations are linear, they can be combined in the leaf (which we will denote by $t_l$). \newline
\indent Utilizing the nomenclature of \cite{algo}, Strassen's formula can be depicted as such: \newline
For $l = 0\cdots6$,
\begin{equation} \label{10}
t_l = \sum_{i,j=0,1}x_{ij}SX(l,i,j) \times \sum_{i,j=0,1}y_{ij}SY(l,i,j),
\end{equation}
and
\begin{equation} \label{11}
q_{ij} = \sum_{l=0}^{6}t_lSQ(l,i,j),
\end{equation}
in which
\begin{center}
$SX$ = 
\begin{tabular}{ | l || c | c | c || r | }
  \textbf{l\textbackslash ij} & \textbf{00} & \textbf{01} & \textbf{10} & \textbf{11} \\
  \textbf{0} & $(1)_{RNS}$ & $(0)_{RNS}$ & $(0)_{RNS}$ & $(0)_{RNS}$\\
  \textbf{1} & $(0)_{RNS}$ & $(1)_{RNS}$ & $(0)_{RNS}$ & $(0)_{RNS}$\\
  \textbf{2} & $(0)_{RNS}$ & $(0)_{RNS}$ & $(1)_{RNS}$ & $(1)_{RNS}$\\
  \textbf{3} & $(-1)_{RNS}$ & $(0)_{RNS}$ & $(1)_{RNS}$ & $(1)_{RNS}$\\
  \textbf{4} & $(1)_{RNS}$ & $(0)_{RNS}$ & $(-1)_{RNS}$ & $(0)_{RNS}$\\
  \textbf{5} & $(0)_{RNS}$ & $(0)_{RNS}$ & $(1)_{RNS}$ & $(1)_{RNS}$\\
  \textbf{6} & $(0)_{RNS}$ & $(0)_{RNS}$ & $(0)_{RNS}$ & $(1)_{RNS}$\\
\end{tabular}
\end{center}
\begin{center}
$SY$ = 
\begin{tabular}{ | l || c | c | c || r | }
  \textbf{l\textbackslash ij} & \textbf{00} & \textbf{01} & \textbf{10} & \textbf{11} \\
  \textbf{0} & $(1)_{RNS}$ & $(0)_{RNS}$ & $(0)_{RNS}$ & $(0)_{RNS}$\\
  \textbf{1} & $(0)_{RNS}$ & $(0)_{RNS}$ & $(1)_{RNS}$ & $(0)_{RNS}$\\
  \textbf{2} & $(-1)_{RNS}$ & $(1)_{RNS}$ & $(0)_{RNS}$ & $(0)_{RNS}$\\
  \textbf{3} & $(1)_{RNS}$ & $(-1)_{RNS}$ & $(0)_{RNS}$ & $(1)_{RNS}$\\
  \textbf{4} & $(0)_{RNS}$ & $(-1)_{RNS}$ & $(0)_{RNS}$ & $(1)_{RNS}$\\
  \textbf{5} & $(0)_{RNS}$ & $(1)_{RNS}$ & $(0)_{RNS}$ & $(1)_{RNS}$\\
  \textbf{6} & $(-1)_{RNS}$ & $(1)_{RNS}$ & $(1)_{RNS}$ & $(-1)_{RNS}$\\
\end{tabular}
\end{center}
\begin{center}
$SQ$ = 
\begin{tabular}{ | l || c | c | c || r | }
  \textbf{l\textbackslash ij} & \textbf{00} & \textbf{01} & \textbf{10} & \textbf{11} \\
  \textbf{0} & $(1)_{RNS}$ & $(1)_{RNS}$ & $(1)_{RNS}$ & $(1)_{RNS}$\\
  \textbf{1} & $(1)_{RNS}$ & $(0)_{RNS}$ & $(0)_{RNS}$ & $(0)_{RNS}$\\
  \textbf{2} & $(0)_{RNS}$ & $(1)_{RNS}$ & $(0)_{RNS}$ & $(0)_{RNS}$\\
  \textbf{3} & $(0)_{RNS}$ & $(1)_{RNS}$ & $(1)_{RNS}$ & $(1)_{RNS}$\\
  \textbf{4} & $(0)_{RNS}$ & $(0)_{RNS}$ & $(0)_{RNS}$ & $(1)_{RNS}$\\
  \textbf{5} & $(0)_{RNS}$ & $(1)_{RNS}$ & $(0)_{RNS}$ & $(0)_{RNS}$\\
  \textbf{6} & $(0)_{RNS}$ & $(0)_{RNS}$ & $(0)_{RNS}$ & $(1)_{RNS}$\\
\end{tabular}
\end{center}
At recursion level $r$, $t_l$ can be represented as such: \newline
For $l = 0\cdots{7^k}-1$,
\begin{equation} \label{12}
t_l = \sum_{i,j = n-1}x_{ij}SX_k(l,i,j) \times \sum_{i,j=0,n-1}y_{ij}SY_k(l,i,j),
\end{equation}
and
\begin{equation} \label{13}
q_{ij} = \sum_{l=0}^{{7^k}-1}t_lSQ_k(l,i,j).
\end{equation}
\indent It is easy to see that $SX = SX_1$, $SY = SY_1$, and $SQ = SQ_1$; however, the difficulty that arises is to determine the values of matrices $SX_k$, $SY_k$, and $SQ_k$ in order to have a \emph{general} algorithm. The following relations were proved in \cite{algo2}, and we shall prove that these results hold with DNA: \newline
\begin{equation} \label{14}
SX_k(l,i,j) = \prod_{r=1}^{k}SX(l_r,i_r,j_r),
\end{equation}
\begin{equation} \label{15}
SY_k(l,i,j) = \prod_{r=1}^{k}SY(l_r,i_r,j_r),
\end{equation}
\begin{equation} \label{16}
SQ_k(l,i,j) = \prod_{r=1}^{k}SQ(l_r,i_r,j_r).
\end{equation}
First we shall extend the definition of the tensor product for arrays of arbitrary dimensions \cite{algo2} by representing the tensor product in RNS by way of single DNA strands.
\begin{proposition}
Let $A$ and $B$ be arrays of the same dimension $l$ and of size $m_1 \times m_2 \times \cdots \times m_l$ and $n_1 \times n_2 \times \cdots \times n_l$, respectively. The elements of $A$ and $B$ are represented using RNA by way of DNA strands as presented in detail in \S 3.1. The tensor product can thus be described as an array of the same dimension and of size $m_1n_1 \times m_2n_2 \times \cdots \times m_ln_l$ in which each element of $A$ is replaced with the product of the element and $B$. This product can be computed with the algorithm $\text{RNSMult}$ which is recognized by a serial of operations of the $\text{RNSAdd}$ algorithm detailed in \S 4.4 of Zheng et al. \cite{zheng}. $P = A \otimes B$ where $P[i_1,i_2,\cdots,i_l] = A[k_1,k_2,\cdots,k_l]B[h_1,h_2,\cdots,h_l]$. $1 \le \forall j \le l$, $i_j = k_jn_j+h_j$ ($k_jn_j$ and $h_j$ will be added with $\text{RNSAdd}$).
\end{proposition}
If we let $P = \otimes_{i=1}^{n}A_i = (\cdots(A_1\otimes A_2)\otimes A_3)\cdots\otimes A_n)$ where $A_i$ is an array of dimension $l$ and of size $m_{i1}\times m_{i2} \times \cdots \times m_{il}$, the following theorem allows us to directly compute the elements of $P$. \emph{All} products and sums of elements can be computed with $\emph{RNSMult}$ and $\emph{RNSAdd}$, respectively.
\begin{theorem} \label{theorem1}
If we let $j_k = \sum_{s=1}^{n}\left(h_{sk}\prod_{r=s+1}^{n}m_{rk}\right)$, then $P[j_1,j_2,\cdots,j_l] = \prod_{i=1}^{n}A_i[h_{i1},h_{i2},\cdots,h_{il}]$.
\end{theorem}
\begin{proof}
We give a proof by induction. For $n = 1$ and $n = 2$, the statement is true. Assume it is true with $n$, then we shall prove that it is true with $n+1$. \newline
\indent $P_{n+1}[v_1,v_2,\cdots,v_l] = \prod_{i=1}^{n+1}A_i[h_{i1},h_{i2},\cdots,h_{il}]$ where
\begin{equation*}
v_k = \sum_{s=1}^{n+1}\left(h_{sk}\prod_{r=s+1}^{n+1}m_{rk}\right),
\end{equation*}
for $1 \le \forall k \le l$. Hence, $P_{n+1} = P_{n}\otimes A_{n+1}$. \newline
Furthermore, by definition,
\begin{equation*}
P_{n+1}[j_1,j_2,\cdots,j_l] = P_{n}[p_1,p_2,\cdots,p_l]A_{n+1}[h_{(n+1)},h_{2(n+1)},\cdots,h_{l(n+1)}] = \prod_{i=1}^{n+1}A_i[h_{i1},h_{i2},\cdots,h_{il}],
\end{equation*}
where
\begin{equation*}
j_k = \sum_{s=1}^{n}\left(h_{sk}\prod_{r=s+1}^{n+1}m_{rk}\right) + h_{k(n+1)} = \sum_{s=1}^{n+1}\left(h_{sk}\prod_{r=s+1}^{n+1}m_{rk}\right).
\end{equation*}
\end{proof}
\begin{theorem} \label{theorem2}
$SX_k = \otimes_{i=1}^{k}SX$, $SY_k = \otimes_{i=1}^{k}SY$, and $SQ_k = \otimes_{i=1}^{k}SQ$.
\end{theorem}
\begin{proof}
We give a proof by induction. For $k = 1$, the statement is true. Assume it is true with $k$, then we shall prove that it is true with $k+1$. \newline
\indent According to \eqref{12} and \eqref{13}, at level $k+1$ of the execution tree, for $0 \le l \le 7^{k+1}-1$
\begin{equation*}
T_l = \left(\sum_{i \ge 0, j \le 2^{k+1}-1}X_{k+1,ij}SX_{k+1}(l,i,j)\right) \times \left(\sum_{i \ge 0, j \le 2^{k+1}-1}Y_{k+1,ij}SY_{k+1}(l,i,j)\right).
\end{equation*}
It follows from \eqref{10} and \eqref{11} that at level $k+2$, for $0 \le l \le 7^{k+1}-1$ and $0 \le l' \le 6$,
\begin{equation} \label{17}
\begin{split}
T_l[l'] = \sum_{i' \ge 0, j' \le 1}\left(\sum_{i \ge 0, j \le 2^{k+1}-1}X_{k+1,ij}[i',j']SX_{k+1}(l,i,j)SX(l',i',j')\right) \times \\ 
\sum_{i'\ge 0, j' \le 1}\left(\sum_{i \ge 0, j \le 2^{k+1}-1}Y_{k+1,ij}[i',j']SY_{k+1}(l,i,j)SY(l',i',j')\right),
\end{split}
\end{equation}
where $X_{k+1,ij}[i',j']$ and $Y_{k+1,ij}[i',j']$ are $2^{k+2} \times 2^{k+2}$ matrices obtained by partitioning the matrices $X_{k+1,ij}$ and $Y_{k+1,ij}$ into 4 sub-matrices (we use $i'$ and $j'$ to denote the sub-matrix's quarter). \newline
\indent We represent $l,l'$ in base 7 RNS, and $i,j,i',j'$ in base 2 RNS. Since $X_{k+1,ij}[i',j'] = X_{k+2,ij}[\overline{ii'_2},\overline{jj'_2}]$, then for $0 \le \overline{ll'}_{(7)} \le 7^{k+1}-1$,
\begin{equation} \label{18}
\begin{split}
M[\overline{ll'}_{(7)}] = \left(\sum_{\overline{ii'}_{(2)} \ge 0,\overline{jj'}_{(2)}\le 2^{k+1}-1}X_{k+2}[\overline{ii'}_{(2)},\overline{jj'}_{(2)}]SX_{k+1}(l,i,j)SX(l',i',j')\right) \times \\
\left(\sum_{\overline{ii'}_{(2)} \ge 0,\overline{jj'}_{(2)} \le 2^{k+1}-1}Y_{k+2}[\overline{ii'}_{(2)},\overline{jj'}_{(2)}]SY_{k+1}(l,i,j)SY(l',i'j')\right).
\end{split}
\end{equation}
Moreover, it directly follows from \eqref{12} and \eqref{13} that for $0 \le \overline{ll'}_{(7)} \le 7^{k+1}-1$,
\begin{equation} \label{19}
\begin{split}
M[\overline{ll'}_{(7)}] = \left(\sum_{\overline{ii'}_{(2)} \ge 0,\overline{jj'}_{(2)}\le 2^{k+1}-1}X_{k+2}[\overline{ii'}_{(2)},\overline{jj'}_{(2)}]SX_{k+2}\left(\overline{ll'}_{(7)},\overline{ii'}_{(2)},\overline{jj'}_{(2)}\right)\right) \times \\
\left(\sum_{\overline{ii'}_{(2)} \ge 0,\overline{jj'}_{(2)} \le 2^{k+1}-1}Y_{k+2}[\overline{ii'}_{(2)},\overline{jj'}_{(2)}]SY_{k+2}\left(\overline{ll'}_{(7)},\overline{ii'}_{(2)},\overline{jj'}_{(2)}\right)\right).
\end{split}
\end{equation}
From \eqref{18} and \eqref{19}, we have
\begin{equation*}
SX_{k+2}\left(\overline{ll'_7},\overline{ii'_2},\overline{jj'_2}\right) = SX_{k+1}(l,i,j)SX(l',i'j'),
\end{equation*}
and
\begin{equation*}
SY_{k+2}\left(\overline{ll'_7},\overline{ii'_2},\overline{jj'_2}\right) = SY_{k+1}(l,i,j)SY(l',i'j').
\end{equation*}
Thus,
\begin{equation*}
SX_{k+2} = SX_{k+1}\otimes SX = \otimes_{i=1}^{k+2}SX,
\end{equation*}
\begin{equation*}
SY_{k+2} = SY_{k+1}\otimes SY = \otimes_{i=1}^{k+2}SY,
\end{equation*}
and
\begin{equation*}
SQ_{k+2} = SQ_{k+1}\otimes SQ = \otimes_{i=1}^{k+2}SQ.
\end{equation*}
\end{proof}
From Theorem~\ref{theorem1} and Theorem~\ref{theorem2}, \eqref{14}, \eqref{15}, and \eqref{16} follow. \newline
\indent As a consequence of \eqref{12}-\eqref{16}, we can form the following sub-matrices:
\begin{equation} \label{17}
T_l = \sum_{i,j=0,2^r-1}X_{ij}\left(\prod_{u=1}^{r}SX(l_u,i_u,j_u)\right) \times \sum_{\substack{i,j=0,2^r-1\\l=0\cdots7^r-1}}Y_{ij}\left(\prod_{u=1}^{r}SX(l_u,i_u,j_u)\right).
\end{equation}
As a result of the storage map of sub-matrices to strands presented in \S 4.2, the following sub-matrices can be \emph{locally} determined within each strand, and their product $T_l$ can be computed by the DNA implementation of the Cannon algorithm presented in \S 4.1:
\begin{equation*}
\left(\sum_{\substack{i=0,2^r-1\\j=0,2^r-1}}X_{ij}\left(\prod_{u=1}^{r}SX(l_u,i_u,j_u)\right)\right),
\end{equation*}
and
\begin{equation*}
\left(\sum_{\substack{i=0,2^r-1\\j=0,2^r-1}}Y_{ij}\left(\prod_{u=1}^{r}SY(l_u,i_u,j_u)\right)\right).
\end{equation*}
All of the sub-matrices are added with the $\emph{RNSMatrixAdd}$ algorithm presented in \S 3.2.3. \newline
\indent Lastly, it is important to note that due to \eqref{12}-\eqref{16}, we have derived a method to directly compute the sub-matrix elements of the resultant matrix via the application of matrix additions (using the $\emph{RNSMatrixAdd}$ algorithm of \S 3.2.3) instead of backtracking manually down the recursive execution tree to compute:
\begin{equation} \label{18}
Q_{ij} = \sum_{l=0}^{7^r-1}T_lSQ_r(l,i,j) = \sum_{l=0}^{7^r-1}T_l\left(\prod_{u=1}^{r}SQ(l_u,i_u,j_u)\right).
\end{equation}
\section{Conclusion}
Our general scalable implementation can be used for all of the matrix multiplication algorithms that use fast matrix multiplication algorithms at the top level (between strands) on a DNA computer. Moreover, since the computational complexity of these algorithms decreases when the recursion level $r$ increases, we can now find optimal algorithms for all particular cases. Of course, as mentioned previously in this paper, the current science of DNA computing does not guarantee a perfect implementation of the Strassen algorithm as described herein; for now, these results should be regarded as primarily theoretical in nature.}
\newpage
\bibliographystyle{amsplain}
}}}

\end{document}